\documentclass[sn-mathphys,Numbered,12pt]{sn-jnl}

\usepackage{graphicx}%
\usepackage{multirow}%
\usepackage{amsmath,amssymb,amsfonts}%
\usepackage{amsthm}%
\usepackage{mathrsfs}%
\usepackage[title]{appendix}%
\usepackage{xcolor}%
\usepackage{textcomp}%
\usepackage{manyfoot}%
\usepackage{booktabs}%
\usepackage{algorithm}%
\usepackage{algorithmicx}%
\usepackage{algpseudocode}%
\usepackage{listings}%

\theoremstyle{thmstyleone}%
\newtheorem{theorem}{Theorem}
\newtheorem{proposition}[theorem]{Proposition}%

\theoremstyle{thmstyletwo}%
\newtheorem{lemma}{Lemma}%
\newtheorem{corollary}{Corollary}%
\newtheorem{example}{Example}%
\newtheorem{remark}{Remark}%
\newtheorem{notation}{Notation}%

\theoremstyle{thmstylethree}%
\newtheorem{definition}{Definition}%
\newtheorem{note}{Note}%

\begin{document}
	
	\title[Article Title]{On the distance distributions of single-orbit cyclic subspace codes}
	
	
	\author{\fnm{Mahak} }\email{mahak@ma.iitr.ac.in}
	
	\author*{\fnm{Maheshanand} \sur{Bhaintwal}}\email{maheshanand@ma.iitr.ac.in}
	
	
	\affil{\orgdiv{Department of Mathematics}, \orgname{Indian Institute of Technology Roorkee}, \orgaddress{\street{} \city{Roorkee}, \postcode{247667}, \state{Uttarakhand}, \country{India}}}
	
	%
	
	
	\abstract{A cyclic subspace code is a union of the orbits of  subspaces contained in it.  In a recent paper, Gluesing-Luerssen et al. (Des. Codes Cryptogr. 89, 447--470, 2021) showed that the study of the distance distribution of a single orbit cyclic subspace code is equivalent to the study of its  intersection distribution. In this paper we have proved that in the orbit of a subspace $U$ of $\mathbb{F}_{q^n}$ that has  the stabilizer $\mathbb{F}_{q^t}^*(t \neq n)$, the number of codeword pairs $(U,\alpha U)$ such that $\dim(U\cap \alpha U)=i$ for any $i,~ 0\leq i < \dim(U)$, is a multiple of $q^t(q^t+1)$, if $\frac{n}{t}$ is an odd number. In the case of even  $\frac{n}{t}$, if  $U$ contains $\frac{q^{2tm}-1}{q^{2t}-1}~ (m\geq 0)$ distinct cyclic shifts of $\mathbb{F}_{q^{2t}}$, then the  number of codeword pairs  $(U,\alpha U)$ with intersection dimension $2tm$ is equal to  $q^t+rq^t(q^t+1)$, for some non-negative integer $r$; and the number of codeword pairs $(U,\alpha U)$ with intersection dimension $i,~(i\neq 2tm)$ is a multiple of $q^t(q^t+1)$. Some examples have been given to illustrate the results presented in the paper.}

	
	
	
	\maketitle
	
	\section{Introduction}\label{sec1}
	
	Subspace codes were first introduced by K\"{o}tter and Kschischang in 2008 \cite{koetter}. They showed how subspace codes may be used for correcting errors and erasures in random network coding. As a result of their work, an enormous amount of research has been done on subspace codes. In recent years, there has been a growing interest in cyclic subspace codes, which are a class of subspace codes, among researchers. These codes were introduced by Etzion and Vardy \cite{etzion}. One of the motivations for studying cyclic subspace codes is their algebraic structure and efficient encoding and decoding algorithms. Several constructions of cyclic subspace codes have been reported in the literature \cite{ben,otal,chen,zhao}.
	
	A subspace code is a collection of $\mathbb{F}_q$-subspaces of $\mathbb{F}_q^n$. Recall that the extension field $\mathbb{F}_{q^n}$  is isomorphic to $\mathbb{F}_q^n$ when considered as a vector space over $\mathbb{F}_q$. Since $\mathbb{F}_{q^n}$ has more algebraic properties than $\mathbb{F}_q^n$, in the study of cyclic subspace codes we consider the subspaces of $\mathbb{F}_q^n$ in $\mathbb{F}_{q^n}$.  A cyclic subspace code $C$ is a collection of $\mathbb{F}_q$-subspaces of $\mathbb{F}_q^n$ that is invariant under the multiplication by the elements of $\mathbb{F}_{q^n}^*$, i.e., for any $\alpha \in \mathbb{F}_{q^n}^*$ and $U\in C$, the subspace $\alpha U=\{\alpha u\mid u \in U\}$ belongs to $C$. The set $\{\alpha U \mid \alpha \in \mathbb{F}_{q^n}^*\}$ can be seen as the orbit of the subspace $U$ under the action of the group $\mathbb{F}_{q^n}^*$ on the set of subspaces of $\mathbb{F}_q^n$ \cite{traut}.
	Thus a cyclic subspace code $C$ is a union of orbits of subspaces contained in $C$.
	
	A cyclic subspace code that contains only one orbit is called a \emph{single-orbit cyclic subspace code} or simply an \emph{orbit code}. Thus, an orbit code is a constant dimension code. If the cardinality of an orbit code is $(q^n-1)/(q-1)$, which is the maximum cardinality of an orbit, then it is called as a full-length orbit code. A full-length orbit code with  subspace distance $2k-2$ is known as an optimal full-length orbit code, where $k$ is the dimension of the subspaces contained in the orbit code.
	
	The authors in \cite{tamo} proposed a novel approach to study cyclic subspace codes. They demonstrated that constructing a single-orbit cyclic subspace code of maximum size $(q^n-1)/(q-1)$ and minimum distance $2k-2$ is equivalent to constructing a Sidon space of dimension $k$ in $\mathbb{F}_q^n$. A \emph{Sidon space} is a subspace $V$ of $\mathbb{F}_q^n$ such that the product of any two non-zero elements of $V$ has a unique factorization over $V$ up to a constant multiple,  i.e., for any $a,b,c,d\in V\backslash \{0\}$ with the property $ab=cd$, implies that $\{a\mathbb{F}_q, b\mathbb{F}_q\}=\{c\mathbb{F}_q, d\mathbb{F}_q\}$.
	
	The \emph{distance distribution} of a code $C$ refers to a vector whose $i$th entry represents the count of pairs of codewords in $C$ that have distance $i$. In \cite{leh}, the authors showed that finding the distance distribution of an orbit code is equivalent to finding its intersection distribution, i.e., counting the number of codeword pairs whose  intersection achieves a specific dimension. They showed that the distance distribution of an optimal full-length orbit code depends solely on the parameters $q,n$, and $k$, irrespective of the choice of the Sidon space. They further investigated the distance distribution of full-length orbit codes with distance $2k-4$. However, they found that, in this case, the distance distribution also depends on the choice of the subspace.
	
	In this paper, we have studied the distance distribution of single-orbit cyclic subspace codes. We have proved that when $n$ is odd, for any subspace $U$ in $\mathbb{F}_{q^n}$ that generates a full-length orbit, the count of codeword pairs $(U,\alpha U)$ in the orbit of $U$  such that $\dim(U\cap \alpha U)=i$ for any $i,~ 0\leq i < \dim(U)$, is always a multiple of $q(q+1)$. 
	For even $n$, if a subspace $U$ contains $\frac{q^{2m}-1}{q^2-1}~(m\geq0)$ number of distinct cyclic shifts of $\mathbb{F}_{q^2}$, then the number of codeword pairs $(U,\alpha U)$ in the orbit of $U$ with intersection dimension $2m$ is equal to $q+ r q(q+1)$ for some non-negative integer $r$, and the number of codeword pairs $(U,\alpha U)$ with intersection dimension $i$ for any $i,~ 0\leq i<\dim(U)~(i\neq 2m)$, is a multiple of $q(q+1)$. We then generalize the result for a single-orbit cyclic subspace code that does not generate a full-length orbit.

	\section{Preliminaries}
	
	We assume throughout that $q$ is a prime power. 
	Let $\mathbb{F}_q$ be a finite field of size $q$, and let $\mathbb{F}_q^* := \mathbb{F}_q\backslash \{0\}$. It is well known that $\mathbb{F}_q^n$ is a vector space of dimension $n$ over $\mathbb{F}_q$. For any subset $\{x_1,x_2,\ldots,x_r\}$ of $\mathbb{F}_q^n$ the
	subspace of $\mathbb{F}_q^n$ spanned by this set is denoted by $\langle x_1,x_2,\ldots,x_r\rangle_{\mathbb{F}_q}$. The set of all $\mathbb{F}_q$-subspaces of $\mathbb{F}_q^n$ is called the projective space of order $n$ over $\mathbb{F}_q$ and is denoted by $\mathcal{P}_q(n)$. For $0\leq k\leq n$, the set of all $k$-dimensional $\mathbb{F}_q$-subspaces of $\mathbb{F}_q^n$ is called a Grassmanian and it is denoted by $\mathcal{G}_q(n,k)$. Clearly, 
	\[\mathcal{P}_q(n)= \bigcup_{k=0}^{n}\mathcal{G}_q(n,k)~.\]
	The size of $\mathcal{G}_q(n,k)$ is given by the $q$-binomial coefficient ${n \brack k}_q$, i.e., 
	\[ |\mathcal{G}_q(n,k)|=  {n \brack k}_q=  \frac{(q^n-1) (q^n-q)\ldots(q^n-q^{k-1})} {(q^k-1) (q^k-q)\ldots(q^k-q^{k-1})}~.\]
	
	A metric  $d_s$, called the \emph{subspace distance}, is defined on $\mathcal{P}_q(n)$ as follows. For any $U, V \in \mathcal{P}_q(n)$,
	\[d_s(U,V)=\dim(U)+\dim(V)-2 \dim(U\cap V)~.\]

	The projective space $\mathcal{P}_q(n)$ is therefore a metric space with respect to the metric $d_s$. A subspace code $C$ is a non-empty subset of $\mathcal{P}_q(n)$ with the metric $d_s$. If every element in a subspace code $C$ is of the same dimension, i.e., $C\subseteq \mathcal{G}_q(n,k)$ for some positive integer $k\leq n$, then we say that $C$ is a constant dimension subspace code. The minimum distance of a subspace code $C$, denoted by $d_s(C)$, is defined by
	\[d_s(C) = \mbox{min}\{d_s(U,V)\mid U,V \in C,~ U\neq V\}~.\]
	
	
	Let $\mathbb{F}_{q^n}$ denote the extension field of degree $n$ of $\mathbb{F}_q$, and let $\mathbb{F}_{q^n}^*\triangleq \mathbb{F}_{q^n} \backslash \{0\}$. Recall that the extension field $\mathbb{F}_{q^n}$ is a vector space of dimension $n$ over $\mathbb{F}_q$ and is isomorphic to $\mathbb{F}_q^n$ as a vector space over $\mathbb{F}_q$. Therefore we will no longer distinguish between $\mathbb{F}_{q^n}$ and $\mathbb{F}_q^n$ as vector spaces over $\mathbb{F}_q$. For a subspace $U\subseteq \mathbb{F}_{q^n} ~\mbox{and} ~\alpha \in \mathbb{F}_{q^n}^*$, the cyclic  shift of $U$ with respect to $\alpha$ is defined as $\alpha U=\{\alpha u \mid u \in U \}$. Clearly $\alpha U \subseteq \mathbb{F}_{q^n}$. It is easy to see that $\alpha U$ is a vector space over $\mathbb{F}_q$ and its dimension is same as the dimension of $U$ over $\mathbb{F}_q$. In fact, we can define a group action $\mathbb{F}_{q^n}^* \times \mathcal{P}_q(n) \rightarrow \mathcal{P}_q(n)$ of $\mathbb{F}_{q^n}^*$  on $\mathcal{P}_q(n)$ \cite{traut} as 
	\begin{eqnarray*}
		(\alpha , U) &\rightarrow & \alpha U~.
	\end{eqnarray*}

	For any $\mathbb{F}_q$-subspace  $U\subseteq\mathbb{F}_{q^n}$, the \emph{orbit} of $U$, denoted by $\mbox{Orb}(U)$, is defined by
	\[\mbox{Orb}(U)=\{\alpha U \mid \alpha \in  \mathbb{F}_{q^n}^* \}~.\]
	The \emph{stabilizer} of $U$, denoted by $\mbox{Stab}(U)$, is defined by $\mbox{Stab}(U)=\{\alpha \in \mathbb{F}_{q^n}^*\mid \alpha U=U\}$. Clearly, $\mbox{Stab}(U)$ is a subgroup of $\mathbb{F}_{q^n}^*$, and since $aU=U$ for all $a \in \mathbb{F}_q^\ast$, we have $\mathbb{F}_q^\ast \subseteq \mbox{Stab}(U)$. By  \cite[Lemma 3.3]{glue}, $\mbox{Stab}(U) \cup \{0\}$ is a subfield of $\mathbb{F}_{q^n}$, and $U$ is a vector space over $\mbox{Stab}(U)\cup\{0\}$. Thus $\mbox{Stab}(U)\cup \{0\}= \mathbb{F}_{q^t}$ for some $t$ which is a divisor of $\mbox{gcd}(\dim_{\mathbb{F}_q}(U),n)$. For $t=n$, $\mbox{Stab}(U)=\mathbb{F}_{q^n}^*$ and $U=\mathbb{F}_{q^n}$. Thus in this case, $\mbox{Orb}(U)$ contains only one element. So, we always consider $t<n$. By the orbit-stabilizer theorem, for any subspace $U$ of $\mathbb{F}_q^n$, we have
	\[\lvert \mbox{Orb}(U)\rvert = \frac{q^n-1}{\lvert \mbox{Stab}(U)\rvert} = \frac{q^n-1}{q^t-1}~.\]
	A subspace code $C$ is called a \emph{cyclic subspace} code if it is closed under cyclic shifts, i.e., for any $\alpha \in \mathbb{F}_{q^n}^*~\mbox{and}~U\in C,~\alpha U\in C$. Clearly, $\mbox{Orb}(U)$ is a cyclic subspace code of constant dimension. We call $\mbox{Orb}(U)$ a \emph{single-orbit cyclic subspace code} or simply an \emph{orbit code}. In general, a cyclic subspace code is the union of the orbits of the subspaces contained in it.

	If $\mbox{Stab}(U)=\mathbb{F}_q^*$, i.e., $\lvert{\mbox{Orb}(U)}\rvert= \frac{q^n-1}{q-1}$, then $\mbox{Orb}(U)$ is called a \emph{full-length orbit code} and we say that $U$ generates a full-length orbit. Otherwise, $\mbox{Orb}(U)$ is a degenerate orbit.
	
	In the next result we show that the dimensions of $U\cap \alpha U$ and $U\cap \alpha^{-1}U$ are same. 
	\begin{lemma}\label{lemma2.1}
		Let  $U$ be a subspace of $\mathbb{F}_{q^n}$ and $\alpha \in \mathbb{F}_{q^n}^*$. Then the dimensions of $U\cap \alpha U$ and $U\cap \alpha^{-1} U$ over $\mathbb{F}_q$ are equal.
	\end{lemma}
	\begin{proof}
		Let $\dim(U\cap \alpha U)=t$, and let $\{v_1,v_2,\ldots, v_t\}$ be a basis of $U\cap \alpha U$. It is easy to see that $\{\alpha^{-1}v_1, \alpha^{-1}v_2,\ldots, \alpha^{-1}v_t\}\subseteq U\cap \alpha^{-1}U$. We will show that $\{\alpha^{-1}v_1, \alpha^{-1}v_2,\ldots, \alpha^{-1}v_t\}$ is a linearly independent set over $\mathbb{F}_q$. Let $\lambda_1,\lambda_2,\ldots, \lambda_t$ be scalars in $\mathbb{F}_q$ such that 
		\[\lambda_1\alpha^{-1}v_1+\lambda_2\alpha^{-1}v_2+\cdots +\lambda_t\alpha^{-1}v_t=0~.\]
		Since $\alpha^{-1}\neq 0 $, it follows that $\lambda_1v_1+\lambda_2v_2+\cdots+\lambda_tv_t=0$. As $\{v_1,v_2,\ldots,v_t\}$ is a linearly independent set over $\mathbb{F}_q$, we get $\lambda_i=0~\mbox{for all}~i,~ 1\leq i\leq t$. Thus $\dim(U\cap \alpha^{-1}U)\geq t$.  Now using $\alpha^{-1}$, instead of $\alpha$, in this inequality, we get  $\dim(U\cap (\alpha^{-1})^{-1}U)\geq \dim(U\cap \alpha^{-1}U)$. Hence $\dim(U\cap \alpha U)=\dim(U\cap \alpha^{-1}U)$.
	\end{proof}

	Let $U$ be an $\mathbb{F}_q$-subspace of dimension $k$ in $\mathbb{F}_{q^n}$. Then $\mbox{Orb}(U)\subseteq \mathcal{G}_q(n,k)$. By the definition of subspace distance for any $\alpha, \beta \in \mathbb{F}_{q^n}^*$,  we have 
	\begin{eqnarray}\label{eq1}
		\notag d_s(\alpha U, \beta U)&=&\dim(\alpha U)+\dim(\beta U)-2 \dim(\alpha U \cap \beta U)\\
		&=& 2k - 2 \dim(\alpha U \cap \beta U)~.
	\end{eqnarray}
	By Lemma \ref{lemma2.1}, $\dim(\alpha U\cap\beta U)= \dim(U\cap\alpha^{-1} \beta U)$, so we get
	\[d_s(\alpha U, \beta U)= 2k- 2\dim(U\cap \alpha^{-1}\beta U) ~.\]
	As a result, the distance distribution of $\mbox{Orb}(U)$ can be determined by finding the dimension of the intersection of $U$ with its cyclic shifts.\\
	Therefore,
	\begin{eqnarray*}
		d_s(\mbox{Orb}(U))&=& \mbox{min}\{d_s(\alpha U,\beta U)\mid \alpha ,\beta \in \mathbb{F}_{q^n}^*, \alpha U \neq \beta U \}\\
		&=& 2k- 2~\mbox{max}\{\dim(U\cap\gamma U)\mid \gamma \in \mathbb{F}_{q^n}^*, \gamma U \neq U\}~.
	\end{eqnarray*}
	
	For any subspace code $C ~\mbox{in}~\mathcal{P}_q(n)$, the \emph{complementary code} $C^\perp$ of $C$ is defined by
	$C^{\bot}={\{V^\bot \mid V\in C\}}$, where $V^{\bot}$ is the orthogonal complement of $V.$ We have $\dim(V^\bot)=n-\dim(V)$, and   $d_s(C)=d_s(C^\bot)$ {\cite{koetter}}.
	
	For any subspaces $U,V$ in $\mathcal{P}_q(n),~ d_s(U^{\bot},V^{\bot})=d_s(U,V).$ From this follows that the distance distributions of a subspace code $C$ and its complementary code are the same. As $\dim(V^\bot)=n- \dim(V)$ for any subspace $V$ of $\mathbb{F}_q^n$, without loss of generality, we may consider only the orbits of subspaces of dimension less than or equal to $n/2$.

	\begin{definition}\cite{leh}
		Let $k\leq \frac{n}{2}$ and let $U$ be an $\mathbb{F}_q$-subspace of dimension $k$ in $\mathbb{F}_{q^n}$. For $0\leq i\leq 2k$, we define 
		\[\delta_i=\lvert\{\alpha U\in \mbox{Orb}(U)\mid d_s(U,\alpha U)=i\}\rvert~.\]
		We call $(\delta_0, \delta_1,\ldots,\delta_{2k})$ the distance distribution of $\mbox{Orb}(U)$. From (\ref{eq1}), we see that the distance between codeword pairs is always even. If $d_s(\mbox{Orb}(U))= 2d$, the non-trivial part of the distance distribution is $(\delta_{2d}, \delta_{2d+2},\ldots, \delta_{2k})$. 
	\end{definition}
\begin{note}	
Here, in the distance distribution, the distance is calculated from the reference space $U$. It is easy to see that for a subspace $U$ in $\mathbb{F}_{q^n}$ with the stabilizer $\mathbb{F}_{q^t}^*$, the number of the codeword pairs of type $(\alpha U, \beta U)$, $\alpha,\beta\in \mathbb{F}_{q^n}^*$
such that $d_s(\alpha U , \beta U)=\delta_{2i}$ is equal to  $\left(\frac{q^n-1}{q^t-1}\right)\delta_{2i}$.
\end{note}
	\begin{definition}\cite{leh}
		Let $k\leq n/2$, and let $U$ be an $\mathbb{F}_q$-subspace of dimension $k$ in $\mathbb{F}_{q^n}$. 
		For $i=0,1,\ldots, k-1,$ we define $\lambda_i=\lvert\mathcal{O}_i(U)\rvert$, where 
		\[\mathcal{O}_i(U)=\{\alpha U\in \mbox{Orb(U)} \mid \dim(U\cap \alpha U)=i\}~.\]
	\end{definition}
	That is, $\lambda_i$ is the number of subspaces $\alpha U$ such that $\dim(U\cap \alpha U)=i$. 
	Suppose $d_s(\mbox{Orb}(U))=2k-2l$. Then $l=\mbox{max}\{\dim(U\cap \alpha U)\mid \alpha\in \mathbb{F}_{q^n}^*, \alpha U\neq U\}$. We call $l$ the maximum intersection dimension of $\mbox{Orb}(U)$, and  $(\lambda_0,\lambda_1, \ldots, \lambda_l)$ the \emph{intersection distribution} of $\mbox{Orb}(U)$.

	\section{Intersection distribution}

	In this section we give some results related to the intersection distribution of orbit codes.

	\begin{lemma}\label{lemma3.1}
		Let $U$ be a subspace in $\mathcal{G}_q(n,k)$ such that $U$ generates a full-length orbit. Then for any $\alpha \in \mathbb{F}_{q^n}\backslash \mathbb{F}_q~\mbox{and}~s\in \mathbb{F}_q^*,~ \dim(U\cap (\alpha+s)U) =\dim(U\cap \alpha U)$.
		
	\end{lemma}
	\begin{proof}
		Let $\alpha\in \mathbb{F}_{q^n}\backslash \mathbb{F}_q$, and let $\dim(U\cap\alpha U)=t.$ Let $s\in \mathbb{F}_q^*$. We claim that $\dim(U\cap (\alpha+s)U)=t$.
		Let $\{u_1,u_2,\ldots,u_t\}$ be a basis of $U\cap\alpha U.$ Then for each $i,~1\leq i\leq t,$ there exists a $v_i\in U$ such that $u_i=\alpha v_i,$ which implies that $u_i+s v_i= \alpha v_i+s v_i= (\alpha+s)v_i$. Let $u_i+s v_i = w_i$, $1\leq i\leq t$. Then $w_i=u_i+s v_i=(\alpha+s)v_i\in U\cap(\alpha+s)U.$ The set $\{u_1=\alpha v_1, u_2=\alpha v_2,\ldots, u_t=\alpha v_t\}$ is linearly independent over $\mathbb{F}_q,$ which implies that $\{w_1=(\alpha+s)v_1,w_2=(\alpha+s)v_2,\ldots, w_t=(\alpha+s)v_t\}$ is linearly independent over $\mathbb{F}_q$. Therefore $\dim(U\cap(\alpha+s)U)\geq \dim(U\cap\alpha U)$.
		
		Similarly, $\dim(U\cap(\alpha+s+(-s)) U)\geq \dim(U\cap (\alpha+s)U),$ which implies that $\dim(U\cap\alpha U)\geq\dim(U\cap(\alpha+s)U)$ and hence, $\dim(U\cap(\alpha+s)U)=\dim(U\cap\alpha U)=t.$
	\end{proof}

	\begin{proposition}\label{prop3.1}
		Let $U$ be a subspace in $\mathcal{G}_q(n,k)$ such that $U$ generates a full-length orbit. Then, for any $\alpha\in \mathbb{F}_{q^n}\backslash \mathbb{F}_q,~ \lvert\{(\alpha+s)U\mid s\in \mathbb{F}_q\}\rvert =q$. 
		
	\end{proposition}
	\begin{proof}	
		Let $\alpha \in \mathbb{F}_q^n\backslash \mathbb{F}_q$. Suppose that $(\alpha+x)U=(\alpha+y)U$ for some $x,y\in \mathbb{F}_q^*$ with $x\neq y$. Then $(\alpha+x)(\alpha+y)^{-1}U=U$.  Since $U$ generates a full-length orbit, $\mbox{Stab}(U)=\mathbb{F}_q^*$. Therefore, $(\alpha+x)(\alpha+y)^{-1} =\delta \in\mathbb{F}_q^*$.
		If $\delta=1,$ then $x=y$, which is not true. Hence $\delta\neq1$, and  $\alpha+x =\delta(\alpha+y)$. This gives $\alpha=\frac{\delta y-x}{1-\delta} \in \mathbb{F}_q$,  which is a contradiction.
		Therefore, for $s\in \mathbb{F}_q, (\alpha+s)U$'s are all distinct subspaces of $\mathbb{F}_q^n$ and hence $\lvert\{(\alpha+s)U\mid s\in \mathbb{F}_q^*\}\rvert =q$.
	\end{proof}
	\begin{corollary}
		Let $U$ be a subspace in $\mathcal{G}_q(n,k)$ such that $U$ generates a full-length orbit. Then, for any $t,~ 0\leq t\leq k-1,$ either $\lambda_t=0$ or $\lambda_t\geq q$, where $\lambda_t$ is as defined above. 
	\end{corollary}
	\begin{proof}
		Let $0\leq t\leq k-1,~\mbox{and}~\lambda_t\neq 0$. Let $\alpha U\in \mathcal O_t(U)$, then $\dim(U \cap \alpha U)=t$. By Lemma \ref{lemma3.1}, $\dim(U\cap (\alpha+s)U)=t$ for all $s\in \mathbb{F}_q^*$, which implies that $\{(\alpha+s)U \mid s\in \mathbb{F}_q\}\subseteq\mathcal O_t(U)$. By Proposition \ref{prop3.1}, $\lvert\{(\alpha+s)U \mid s\in \mathbb{F}_q\}\rvert =q$. Hence, $\lambda_t=\lvert\mathcal O_t(U)\rvert\geq q$. 
	\end{proof}

	\begin{theorem}\label{th3.2}
		Let $U$ be a subspace in $\mathcal{G}_q(n,k)$, and let $U$ generate a full-length orbit. Then, for any $t,~ 0\leq t\leq k-1,~\lambda_t$ is a multiple of $q$.
	\end{theorem}
	
	\begin{proof}
		Let $0\leq t\leq k-1,$ if $\lambda_t=0$ then the result is trivially true. Let $\lambda_t\neq 0,$ and $\alpha U\in \mathcal O_t(U)$. Then $\dim(U\cap \alpha U)=t$, and by Lemma \ref{lemma3.1}, $\dim(U\cap (\alpha+s)U)=t,$ for all $s\in \mathbb{F}_q$. Consider the set $X=\{(\alpha+\lambda) U\mid \lambda\in \mathbb{F}_q \}$. By Proposition \ref{prop3.1}, $\rvert X\lvert=q,$ and $X\subseteq \mathcal O_t(U)$. If $ X=\mathcal O_t(U)$ then the result is true. Suppose there exists a $\beta U\in\mathcal O_t(U)$ such that $\beta U\notin X$. Let $Y=\{(\beta+\delta)U\mid \delta\in \mathbb{F}_q\}$, by Lemma \ref{lemma3.1}, $Y\subseteq\mathcal O_t(U)$ and by Proposition \ref{prop3.1}, $\lvert Y\rvert=q$. 
		
		We will show that $Y\cap X=\emptyset$. 
		Let if possible $(\beta+y)U\in X$ for some $y\in\mathbb{F}_q^*$. Then $(\beta+y)U=(\alpha+x)U~ \mbox{for some}~ x\in \mathbb{F}_q$, i.e., $(\alpha+x)^{-1}(\beta+y)U = U$. Since $U$ generates a full-length orbit, the stabilizer of $U$ under the action $(\alpha, U) \mapsto \alpha U$ is $\mathbb{F}_q^\ast$, where $\alpha \in \mathbb{F}_{q^n}^\ast$. So we have $(\alpha+x)^{-1}(\beta+y) \in \mathbb{F}_q^\ast$, and hence $\beta+y = \delta(\alpha+x)$ for some $\delta\in \mathbb{F}_q^*$. From this we get $\beta =\delta(\alpha+x)-y$.
		
		Now let  $h=x-\delta^{-1}y$. Clearly, $h \in \mathbb{F}_q$. Then $\beta=\delta(\alpha+h)$, and hence  $\beta U=(\alpha+h) U$. Thus, $\beta U\in X$, which is a contradiction. Therefore, $Y\cap X=\emptyset.$ Thus, $\mathcal O_t(U)$ can be written as a disjoint union of sets each with cardinality $q.$ Hence $\lambda_t=\rvert\mathcal O_t(U)\lvert$ is a multiple of $q$.
	\end{proof}
	
	\begin{example}
		Let $q=3$ and  $n=10$. Let $z$ be a primitive element of $\mathbb{F}_{3^{10}}$. Let $x=z^{\frac{3^{10}-1}{3^2-1}}=z^{7381}$. Then the degree of the minimal polynomial of $x$ over $\mathbb{F}_3$ is $2$, and  $\mathbb{F}_{3^2}={\langle 1, x\rangle}_{\mathbb{F}_3}$. Define $U= \beta_1\mathbb{F}_{3^2}\oplus \beta_2\mathbb{F}_{3^2}\oplus a\mathbb{F}_3$, where $\beta_1=z^{1708},~\beta_2=z^{732}~\mbox{and}~a=z^{91}$. The dimension of $U$ over $\mathbb{F}_3$ is $5$. By using the Magma Computational Algebra System \cite{magma}, we obtained $\mbox{Stab}(U)=\mathbb{F}_3^*$, and $\lambda_0=17280,~ \lambda_1=11520,~\lambda_2=720,~\lambda_3=0, \mbox{and}~\lambda_4=3$. Clearly, each $\lambda_i$ is a multiple of $q=3$. 
	\end{example}

	\begin{theorem}\label{thm3.3}
		Let $U$ be a subspace in $\mathcal{G}_q(n,k)$ such that $U$ generates a full-length orbit. If $n$ is odd or $q$ is a power of 2, then for any $\beta \in \mathbb{F}_{q^n}\setminus\mathbb{F}_q$, the subspaces $\beta U$ and $\beta^{-1}U$ are distinct.
\end{theorem}

\begin{proof}
	Let $\beta \in \mathbb{F}_{q^n}\backslash \mathbb{F}_q$. Suppose that $\beta U=\beta^{-1}U.$ Since $U$ generates a full-length orbit, $\beta U=\beta^{-1}U$ implies that $\beta^2\in \mathbb{F}_q^*$.
	
	Let $n$ be an odd number. As $\beta^2\in \mathbb{F}_q^*,$ we have $(\beta^2)^{(q-1)}=1$, which implies that the order of $\beta$ divides $2(q-1)$. Thus, the order of $\beta$ divides $\mbox{gcd}(q^n-1,2(q-1))$. Since $n$ is an odd number, then $\frac{q^n-1}{q-1}=1+q+\cdots+q^{n-1}$ is also an odd number. So, $\mbox{gcd}(q^n-1,2(q-1))=q-1$. Thus the order of $\beta$ divides $q-1$. This implies that $\beta\in \mathbb{F}_q$, which is a contradiction.
	
	If $q$ is a power of $2$,  then for any positive integer $n,~ \frac{q^n-1}{q-1}=1+q+\cdots+q^{n-1}$ is an odd number. By the above argument, we get a contradiction that $\beta \in \mathbb{F}_q$.
\end{proof}

\begin{corollary}
	Let $U$ be a subspace in $\mathcal{G}_q(n,k)$ such that $U$ generates a full-length orbit. If $n$ is odd and $q$ is a power of an odd prime then for each $t,~0\leq t\leq k-1,~\lambda_t=\rvert\mathcal O_t(U)\lvert$ is a multiple of $2q$.
\end{corollary}
\begin{proof}
	Let $n$ be an odd positive integer and let $q$ be a power of an odd prime. If $\lambda_t=0$ then the result is trivially true.  Let $\lambda_t \ne 0$, i.e., $\mathcal O_t(U)\neq \emptyset$. By Lemma \ref{lemma2.1}, $\dim (U \cap \beta U) = \dim (U \cap \beta^{-1} U)$ for any $\beta \in \mathbb{F}_{q^n}^\ast$. Therefore, $\beta^{-1}U \in \mathcal O_t(U)$ for any $\beta U \in \mathcal O_t(U)$. By Theorem \ref{thm3.3}, $\beta U$ and $\beta^{-1} U$  are distinct subspaces. This implies  that $\lambda_t$ is an even number. By Theorem \ref{th3.2}, $\lambda_t$ is a multiple of $q$. Hence $\lambda_t$ is a multiple of $2q$.
\end{proof}

\begin{notation}
	For a subspace $U$ in $\mathcal{G}_q(n,k)$ and $\alpha \in \mathbb{F}_{q^n}\backslash \mathbb{F}_q$, we define 
	\[U_\alpha=\{(\alpha+\delta)U\mid \delta\in \mathbb{F}_q\}~.\]
	Then for any $\lambda\in \mathbb{F}_q$, we have
	\[U_{(\alpha+\lambda)^{-1}}=\{((\alpha+\lambda)^{-1}+\delta)U\mid \delta\in \mathbb{F}_q\}~.\] 
	We define 
	\[\mathcal S_{\alpha,U}=U_\alpha \cup \big(\cup_{\lambda\in \mathbb{F}_q} U_{(\alpha+\lambda)^{-1}}\big)~.\]
\end{notation}

\begin{lemma}\label{lemma3.4}
	Let $U$ be a subspace in $\mathcal{G}_q(n,k),~\mbox{and let}~\alpha, \beta\in \mathbb{F}_{q^n}\backslash\mathbb{F}_q$. Suppose that $U$ generates a full-length orbit. Then $\mathcal S_{\alpha,U}$ and $\mathcal S_{\beta,U}$ are either identical or disjoint. 
\end{lemma}
\begin{proof}
	Suppose that $\mathcal S_{\alpha,U}\cap\mathcal S_{\beta,U} \neq \emptyset$. Let $\gamma U\in \mathcal S_{\alpha,U} \cap \mathcal S_{\beta, U}$.	Then $\gamma U\in \mathcal S_{\alpha,U} = U_\alpha\cup \big(\cup_{\lambda\in \mathbb{F}_q}U_{({\alpha+\lambda})^{-1}}\big)$. \\
	\textbf{Case 1}. Let $\gamma U\in U_{\alpha}$. Then $\gamma U=(\alpha+\lambda)U$ for some $\lambda \in \mathbb{F}_q$. From this we get $\gamma(\alpha+\lambda)^{-1}U=U$. Since $U$ generates a full-length orbit, $\mbox{Stab}(U)=\mathbb{F}_q^*$. This implies that $\gamma(\alpha+\lambda)^{-1}\in \mathbb{F}_q^*$. Then there exists an element $s$ in $\mathbb{F}_q^*$ such that 
	\begin{equation}\label{eq2}
		\gamma = s(\alpha +\lambda)~.
	\end{equation}
	Now for any $\delta\in \mathbb{F}_q$, 
	\begin{equation}\label{eq3}
		\gamma+\delta=s(\alpha+\lambda+s^{-1}\delta)~.
	\end{equation} 
	This implies that $(\gamma+\delta)U=(\alpha+\lambda+s^{-1}\delta)U \in \mathcal S_{\alpha,U}$. Thus $U_{\gamma}\subseteq \mathcal S_{\alpha,U}$. 
	
	From equation (\ref{eq3}), $(\gamma+\delta)^{-1}=s^{-1}(\alpha+\lambda+s^{-1}\delta)^{-1}$. Then for any $x\in \mathbb{F}_q$, $(\gamma+\delta)^{-1}+x=s^{-1}\big((\alpha+\lambda+s^{-1}\delta)^{-1}+sx\big)$. From this we get 
	\[\big((\gamma+\delta)^{-1}+x\big)U=\big((\alpha+\lambda+s^{-1}\delta)^{-1}+sx\big)U~.\]
	As $\big((\alpha+\lambda+s^{-1}\delta)^{-1}+sx\big)U \in U_{(\alpha+\lambda+s^{-1}\delta)^{-1}}$, we get $U_{(\gamma+\delta)^{-1}}\subseteq \mathcal S_{\alpha,U},$ for every $\delta\in \mathbb{F}_q$. Hence, $\mathcal S_{\gamma,U}\subseteq \mathcal S_{\alpha,U}$. 
	
	From equation (\ref{eq2}), $\alpha=s^{-1}(\gamma-s\lambda)$. This implies that $\alpha U =(\gamma -s\lambda)U,~ \mbox{and hence}~\alpha U \in U_{\gamma}$. Therefore, $\mathcal S_{\alpha,U}\subseteq \mathcal S_{\gamma,U}$. Hence $\mathcal S_{\alpha,U}=\mathcal S_{\gamma,U}$.\\
	\textbf{Case 2}. Let $\gamma U \in U_{(\alpha+x)^{-1}}$, for some $x \in \mathbb{F}_q$. Then $\gamma U =\big((\alpha+x)^{-1}+y\big)U$, for some $y\in \mathbb{F}_q$. As $U$ generates a full-length orbit, there exists an element $b$ in $\mathbb{F}_q$ such that 
	\begin{equation}\label{eq4}
		\gamma = b\big((\alpha+x)^{-1}+y\big)~. 
	\end{equation}
	For any $c \in \mathbb{F}_q,~ \gamma+c = b\big((\alpha+x)^{-1}+y+b^{-1}c\big), ~\mbox{and}~ (\gamma+c)U =\big((\alpha+x)^{-1}+y+b^{-1}c\big)U \in U_{(\alpha+x)^{-1}}$. From this, we get $U_{\gamma}\subseteq U_{(\alpha+x)^{-1}}$. So $U_\gamma \subseteq S_{\alpha, U}$.
	
Now by equation (\ref{eq4}), for any $m$ in $\mathbb{F}_q$, $(\gamma+m)^{-1}=b^{-1}\big((\alpha+x)^{-1}+y+b^{-1}m\big)^{-1}$. Let $z=y+b^{-1}m\in \mathbb{F}_q$. So, we get $(\gamma+m)^{-1}=b^{-1}\big((\alpha+x)^{-1}+z\big)^{-1}$. This implies that $(\gamma+m)^{-1}= \frac{b^{-1}(\alpha+x)}{1+z(\alpha+x)}$. From this we get $(\gamma+m)^{-1}= -b^{-1}z^{-2} \big((\alpha+x+z^{-1})^{-1}-z\big)$. Then for any $g$ in $\mathbb{F}_q$, $(\gamma+m)^{-1}+g= -b^{-1}z^{-2} \big((\alpha+x+z^{-1})^{-1}-z-bz^2g\big)$. Thus $((\gamma+m)^{-1}+g)U=\big((\alpha+x+z^{-1})^{-1}-z-bz^2g\big)U\in U_{(\alpha+x+z^{-1})^{-1}}$. This implies that  $U_{(\gamma+m)^{-1}}\subseteq S_{\alpha,U}$ for any $m\in \mathbb{F}_q.$ Since we have $U_\gamma \subseteq S_{\alpha, U}$, we get  $\mathcal S_{\gamma,U}\subseteq \mathcal S_{\alpha,U}$.
	
	From equation (\ref{eq4}), we get $b^{-1}\gamma-y=(\alpha+x)^{-1}$. From it follows that $\alpha+x=b(\gamma-by)^{-1}$, which gives $\alpha=b\big((\gamma-by)^{-1}-b^{-1}x)\big)$. Thus, $\alpha U\in U_{(\gamma-by)^{-1}},~\mbox{and}~ \mathcal S_{\alpha,U}\subseteq \mathcal S_{\gamma,U}$. Hence, we get $\mathcal S_{\alpha,U}=\mathcal S_{\gamma,U}$.
	
	Similarly, $\mathcal S_{\gamma, U}=\mathcal S_{\beta,U}$. From this follows that $\mathcal S_{\alpha,U}~ \mbox{and}~ \mathcal S_{\beta,U}$ are identical. 
\end{proof}

\begin{lemma}\label{lemma3.5}
	Let $U$ be a subspace in $\mathcal G_q(n,k),~\mbox{and let}~\alpha\in \mathbb{F}_q^n\backslash \mathbb{F}_q$. Suppose that $U$ generates a full-length orbit and the degree of $\alpha$ over $\mathbb{F}_q$ is not equal to $2$. Then the cardinality of $\mathcal S_{\alpha,U}$ is $q(q+1)$. 
\end{lemma}

\begin{proof}
	By Proposition \ref{prop3.1},  $\lvert U_{\alpha}\rvert=\lvert U_{(\alpha+\lambda)^{-1}}\rvert=q,~\mbox{for all}~ \lambda\in \mathbb{F}_q$. We claim that all the sets $U_{\alpha}, U_{(\alpha+\lambda)^{-1}}, \lambda\in \mathbb{F}_q$, are pairwise disjoint. 
	
	First suppose that $U_\alpha\cap \big(\cup_{\lambda\in \mathbb{F}_q}U_{({\alpha+\lambda})^{-1}}\big)\neq \emptyset$. Then there exist $x,y,\delta \in \mathbb{F}_q$ such that \[(\alpha+x)U=((\alpha+y)^{-1}+\delta)U~.\]
	As $U$ generates a full-length orbit, $\frac{\alpha+x}{(\alpha+y)^{-1}+ \delta}\in \mathbb{F}_q^*$. Then $\alpha+x= \lambda ((\alpha+y)^{-1}+ \delta)$ for some $\lambda \in \mathbb{F}_q^*$. This gives $\alpha^2+ (x+y-\lambda \delta)\alpha+(x-\lambda\delta)y-\lambda=0$. Let $b=x+y-\lambda \delta~\mbox{and}~c=(x-\lambda\delta)y-\lambda$, then we have $\alpha^2+b\alpha+c=0,~\mbox{where}~b,c \in \mathbb{F}_q$. It means that the degree of the minimal polynomial of $\alpha$ over $\mathbb{F}_q$ is $2$, 
	which is not true. Hence, $U_\alpha\cap \big(\cup_{\lambda\in \mathbb{F}_q}U_{({\alpha+\lambda})^{-1}}\big) = \emptyset$. 
	
	Now suppose that $U_{({\alpha+\lambda})^{-1}}\cap U_{({\alpha+\delta})^{-1}}\neq \emptyset,~\mbox{for some distinct elements}~\lambda,\delta \in \mathbb{F}_q$. Then $((\alpha+\lambda)^{-1}+c)U=((\alpha+\delta)^{-1}+d)U,~\mbox{for some}~c,d\in \mathbb{F}_q$. As $U$ generates a full-length orbit, $\frac{(\alpha+\lambda)^{-1}+c}{(\alpha+\delta)^{-1}+d}\in \mathbb{F}_q^*$. Then there exists an element $s\in \mathbb{F}_q^*$ such that $ (\alpha+\lambda)^{-1}+c= s((\alpha+\delta)^{-1}+d)$. By solving this, we get
	\[(c-ds)(\alpha^2+(\lambda+\delta)\alpha+\lambda\delta)=(s-1)\alpha+(s\lambda-\delta)~.\] 
	If $c-ds=0$, then we get $(s-1)\alpha+(s\lambda-\delta)=0.$ 
	As $\lambda\neq \delta,~ s-1\neq 0.$ So, we get $\alpha\in \mathbb{F}_q$, which is not true. Thus $c-ds\neq 0$. So $\alpha$ is a root of a quadratic equation over $\mathbb{F}_q$. Then the degree of  $\alpha$ over $\mathbb{F}_q$ is $2$, which is not true. Therefore, $U_{({\alpha+\lambda})^{-1}}\cap U_{({\alpha+\delta})^{-1}}= \emptyset$. Hence all the sets $U_{\alpha}, U_{(\alpha+\lambda)^{-1}}, \lambda\in \mathbb{F}_q$, are pairwise disjoint and $\lvert \mathcal S_{\alpha,U}\rvert= q(q+1)$. 
\end{proof}

When the degree of $\alpha$ over $\mathbb{F}_q$ is $2$,  we have the following result for the cardinality of $\mathcal S_{\alpha,U}$.

\begin{lemma}\label{lemma3.7}
	Let $U$ be a subspace in $\mathcal G_q(n,k),~\mbox{and let}~\alpha\in \mathbb{F}_q^n\backslash \mathbb{F}_q$. Suppose that $U$ generates a full-length orbit and the degree  of $\alpha$ over $\mathbb{F}_q$ is equal to $2$. Then the cardinality of $\mathcal S_{\alpha,U}$ is $q$.
\end{lemma}

\begin{proof}
	We claim that all the sets $U_{\alpha}, U_{(\alpha+\lambda)^{-1}}, \lambda\in \mathbb{F}_q$, are identical. Since the degree of the minimal polynomial of $\alpha$ over $\mathbb{F}_q$ is $2$, there exist elements $b,c~\mbox{in}~ \mathbb{F}_q$ such that 
	\[\alpha^{2}+b\alpha+c=0~. \]
	Then for any $\lambda\in \mathbb{F}_q$, we can rephrase the above equation as
	\begin{equation}\label{eq5}
		(\alpha+\lambda) (\alpha+b-\lambda)=-c+\lambda(b-\lambda)~.
	\end{equation}
	Now if $\lambda(b-\lambda)-c=0$ then $\alpha+\lambda=0$ or $\alpha+b-\lambda=0$. In both the cases, we get $\alpha \in \mathbb{F}_q$, which is not true. Hence $\lambda(b-\lambda)-c\neq 0$.\\
	From equation (\ref{eq5}), we get 
	\[(\alpha+\lambda)^{-1}= (\lambda(b-\lambda)-c)^{-1}(\alpha+b-\lambda)~. \] 
	Now for any $\delta\in \mathbb{F}_q,~(\alpha+\lambda)^{-1}+\delta= (\lambda(b-\lambda)-c)^{-1}\big(\alpha+b-\lambda+(\lambda(b-\lambda)-c)\delta\big)$. Then $\big((\alpha+\lambda)^{-1}+\delta\big)U= \big(\alpha+b-\lambda+(\lambda(b-\lambda)-c)\delta\big)U \in U_\alpha$ as $b-\lambda+(\lambda(b-\lambda)-c)\delta\in \mathbb{F}_q$. So $U_{(\alpha+\lambda)^{-1}}\subseteq U_{\alpha}$. Since $\lambda\in \mathbb{F}_q$ is arbitrary, we conclude that $U_{(\alpha+\lambda)^{-1}}\subseteq U_{\alpha}$ for every $\lambda\in \mathbb{F}_q$. Hence $\mathcal S_{\alpha,U}=U_{\alpha}$. By Proposition \ref{prop3.1}, the cardinality of $U_\alpha$ is $q$. The result follows. 
\end{proof}

\begin{theorem}\label{thm3.7}
	Let $U$ be a subspace in $\mathcal{G}_q(n,k)$ such that $U$ generates a full-length orbit. If $n$ is an odd number then for each $t,~0\leq t\leq k-1,~\lambda_t=\lvert\mathcal O_t(U)\rvert$ is a multiple of $q(q+1)$.
\end{theorem}

\begin{proof}
	If $\lambda_t=0$ then the result is trivially true. Let $\lambda_t\neq 0,$ and let $\alpha U\in \mathcal O_t(U)$. Then $\dim(U\cap \alpha U)=t$. By Lemma \ref{lemma2.1} and Lemma \ref{lemma3.1}, $\dim(U\cap \gamma U)= \dim(U\cap \alpha U)= t$ for any $\gamma U \in \mathcal S_{\alpha,U}$. Therefore, $\mathcal S_{\alpha,U}\subseteq \mathcal{O}_t(U)$. 
	
	We know that for any $x \in \mathbb{F}_{q^n}$, the degree of  $x$ over $\mathbb{F}_q$ divides $n$. Since $n$ is odd, there is no element in $\mathbb{F}_{q^n}$  of  degree $2$. Then by Lemma \ref{lemma3.5}, $\lvert \mathcal S_{\alpha,U}\rvert=q(q+1)$. Now suppose there is a subspace $\beta U \in \mathcal O_t(U)$ such that $\beta U \notin \mathcal S_{\alpha,U}$. Then by Lemma \ref{lemma3.4}, $\mathcal S_{\alpha,U}\cap \mathcal S_{\beta,U}=\emptyset$. Therefore, $\mathcal{O}_t(U)$ can be written as a disjoint union of sets each with cardinality $q(q+1)$. Hence, the cardinality of $\mathcal O_t(U)$ is a multiple of $q(q+1)$.
\end{proof}

\begin{example}
	Let $q=3$ and $n=11$. Let $z$ be a primitive element of $\mathbb{F}_{3^{11}}$. Define $U=\langle z^{13}, z^{17}, z^{21}, z^{23}\rangle_{\mathbb{F}_3}$. Here, $U$ is a vector space of dimension $4$ over $\mathbb{F}_3$. As $\mbox{gcd}(\dim(U), n)= 1,~\mbox{Stab}(U)=\mathbb{F}_3^*$. Thus, $U$ generates a full-length orbit. By using Magma, we have obtained the values of $\lambda_i$. We have $\lambda_0=87048,~ \lambda_1=1512,~ \lambda_2=12~\mbox{and}~\lambda_3=0$. Clearly, each $\lambda_i$ is a multiple of $q(q+1)=12$. 
\end{example}

\begin{lemma}\label{lemma3.8}
	If $n$ is an even positive integer, then
	\[\left\{x \in \mathbb{F}_{q^n}:~
	\deg_{\mathbb{F}_q}(x)=2\right\} = \mathbb{F}_{q^2}\backslash \mathbb{F}_q~.\]
\end{lemma}
\begin{proof}
	Let $n$ be an even positive integer and let $T=\left\{x \in \mathbb{F}_{q^n}:~\deg_{\mathbb{F}_q}(x)=2\right\}$. As $n$ is even,  by \cite[Theorem 6.18]{finite},  $\mathbb{F}_{q^2}$ is the unique subfield of $\mathbb{F}_{q^n}$ with cardinality $q^2$. For any $\alpha\in \mathbb{F}_{q^2}\backslash\mathbb{F}_q$, the degree of the minimal polynomial of $\alpha$ divides $2$. 
	As $\alpha\notin \mathbb{F}_q$, degree of $\alpha$ cannot be  $1$. Thus the degree of  $\alpha$ is $2$. From this follows that $\mathbb{F}_{q^2}\backslash \mathbb{F}_{q}\subseteq T$.
	
	Let $\beta \in T$. Then $\mbox{deg}_{\mathbb{F}_q}(\beta)=2$.  Now, by \cite[Theorem 5.8]{finite} $\mathbb{F}_q(\beta)$ is a subfield of $\mathbb{F}_{q^n}$ and $\lvert \mathbb{F}_q(\beta)\rvert = q^2$. So $\mathbb{F}_q(\beta)=\mathbb{F}_{q^2}$, and $\beta \in \mathbb{F}_{q^2}\backslash\mathbb{F}_q$. Hence $T=\mathbb{F}_{q^2}\backslash \mathbb{F}_q$. 
	\end{proof}
\begin{remark}\label{remark3.1}
	Let $U\in \mathcal{G}_q(n,k)$ be such that $U$ generates a full-length orbit. Let $\gamma$ be an arbitrary element in $\mathbb{F}_{q^2}\backslash\mathbb{F}_q$. Then $\mathbb{F}_{q^2}\backslash \mathbb{F}_q=\{a+b\gamma\mid a, b\in \mathbb{F}_q,~ b\ne 0\}$. 
	For any $x=a+b\gamma \in \mathbb{F}_{q^2}\backslash \mathbb{F}_q,~ xU= (a+b\gamma )U=(\gamma+b^{-1}a)U$. From this, we get $xU \in \mathcal{S}_{\gamma,U}$. Therefore, the cyclic shifts of $U$ of the form $\alpha U$, $\alpha \in \mathbb{F}_{q^2}\backslash \mathbb{F}_q$, in  $\mbox{Orb}(U)$  are contained in $\mathcal{S}_{\gamma,U}$, where $\gamma$ is an arbitrary element of $\mathbb{F}_{q^2}\backslash\mathbb{F}_q$.
\end{remark}

\begin{theorem}\label{thm3.9}
	Let $n$ be an even integer and let $U$ be a subspace in $\mathcal{G}_q(n,k)$. Let $\alpha\in \mathbb{F}_q^n\backslash \mathbb{F}_q$ and $\alpha\notin \mbox{Stab}(U)$. Let the degree of $\alpha$ over $\mathbb{F}_q$ be $2$. Let $V=U \cap \alpha U$ and $V\neq \{0\}$. Then $\mathbb{F}_{q^2}^*\subseteq \mbox{Stab}(V)$.
\end{theorem}

\begin{proof}
	As degree of $\alpha$ over $\mathbb{F}_q$ is $2$, there exist elements $c,d\in \mathbb{F}_q$ such that 
	\begin{equation}\label{eq6}
		\alpha^2+c\alpha+d=0~.
	\end{equation} 
	Since $\alpha \notin \mathbb{F}_q$, we have $d\neq 0$. \\
		\textbf{Case 1}. Let $q=2^s$.
		If $c=0$ then $\alpha^2=-d$. By \cite[Theorem 6.20]{finite}, every element of $\mathbb{F}_{2^s}$ has a unique square root in $\mathbb{F}_{2^s}$. Then $\alpha\in\mathbb{F}_q$, which is not true. So, $c\ne 0$. 
	Now equation (\ref{eq6}) can be rephrased as
\[\big((d^{-1}c)\alpha+1\big)(\alpha+c)=\alpha~,\]
from which we get $\alpha+c =\frac{\alpha}{\big(1+(d^{-1}c)\alpha\big)}$. 
	Now, $\big(1+(d^{-1}c)\alpha\big)V\subseteq V+\alpha V\subseteq \alpha U$. From this follows that $V\subseteq \big(\frac{\alpha}{\big(1+(d^{-1}c)\alpha\big)}\big)U=(\alpha+c)U$. From (\ref{eq6}), $(\alpha^2+c\alpha)U=U$. Therefore, $V\subseteq (\alpha+c)U\cap\big(\alpha(\alpha+c)\big)U=(\alpha+c)V$. Hence $V=(\alpha+c)V$, as $(\alpha+c)V \subseteq V$. This implies that $\alpha+c\in \mbox{Stab}(V)$. The $\mbox{Stab}(V)\cup\{0\}$ is a subfield of $\mathbb{F}_{q^n}$ containing $\mathbb{F}_q^*$. As $c \in \mathbb{F}_q^\ast \subseteq \mbox{Stab}(V)$, it follows that $\alpha\in \mbox{Stab}(V)$. Thus $\mathbb{F}_q(\alpha)\subseteq \mbox{Stab}(V)\cup \{0\}$. As the degree of $\alpha$ over $\mathbb{F}_q$ is $2$, $\mathbb{F}_q(\alpha)=\mathbb{F}_{q^2}$. From this we conclude that $\mathbb{F}_{q^2}\subseteq \mbox{Stab}(V)\cup\{0\}$.
	\textbf{Case 2}. Let $q$ be a power of an odd prime.\\
	If $c=0$ then $\alpha^2 =-d\in \mathbb{F}_q^*$. Since $\alpha^2\in \mathbb{F}_q^*$, $\alpha^2 U=U$ and  $V\subseteq \alpha U\cap \alpha^2U=\alpha V$. So, we get $\alpha V=V$. Then $\alpha \in \mbox{Stab}(V)$. As $\mathbb{F}_q^*\subseteq \mbox{Stab}(V),~\mathbb{F}_{q^2}=\mathbb{F}_q(\alpha)\subseteq \mbox{Stab}(V) \cup\{0\}$. If $c\neq 0$ then by the argument used in case 1, we get $\mathbb{F}_{q^2}\subseteq \mbox{Stab}(V)\cup\{0\}$. 
\end{proof}

\begin{corollary}\label{cor3}
	Let $n$ be an even integer. Let $U$ be a subspace in $\mathcal{G}_q(n,k)$ such that $U$ generates a full-length orbit. Let $\beta\in \mathbb{F}_q^n\backslash \mathbb{F}_q$ be such that the degree of $\beta$ over $\mathbb{F}_q$ is $2$. Then $\dim_{\mathbb{F}_q}(U\cap \beta U)$ is an even number. 
\end{corollary}
\begin{proof}
	Let $V=U\cap \beta U$ and let  $\mbox{Stab}(V)\cup\{0\}=\mathbb{F}_{q^t}$. Then $t$  divides $\mbox{gcd}(\dim_{\mathbb{F}_q}(V),n)$. By Theorem \ref{thm3.9}, $\mathbb{F}_{q^2}\subseteq\mathbb{F}_{q^t}$. Therefore, $t$ is an even number. As $t$ divides  $\mbox{gcd}(\dim_{\mathbb{F}_q}(V),n)$, $\dim_{\mathbb{F}_q}(V)$ is an even number. 	
\end{proof}

\begin{theorem}\label{thm3.11}
	Let $n$ be an even number. Let $U$ be a subspace in $\mathcal{G}_q(n,k)$ such that $U$ generates a full-length orbit. Then for any odd $t,~ 0\leq t\leq k-1,~ \lambda_t=\lvert \mathcal O_t(U) \rvert$ is a multiple of $q(q+1)$. 
\end{theorem}

\begin{proof}
	Let $t$ be an odd number in the range $0\leq t\leq k-1$.  Let $\alpha U\in \mathcal{O}_t(U)$. Then $\dim(U\cap \alpha U)=t$, and by Lemma \ref{lemma2.1} and Lemma \ref{lemma3.1}, $\mathcal{S}_{\alpha,U}\subseteq\mathcal{O}_{t}(U)$. As $\dim (U\cap \alpha U)$ is odd, by Corollary \ref{cor3} the degree of $\alpha$ over $\mathbb{F}_q$ is not equal to $2$. By Lemma \ref{lemma3.5}, the cardinality of $\mathcal{S}_{\alpha,U}$ is $q(q+1)$.
	
	Now suppose that there exists $\beta U \in \mathcal{O}_t(U)$ such that $\beta U \notin \mathcal{S}_{\alpha,U}$. Then $\mathcal{S}_{\beta,U}\subseteq \mathcal{O}_{t}(U)$ and by Lemma \ref{lemma3.4}, $\mathcal{S}_{\alpha,U}\cap\mathcal{S}_{\beta,U}=\emptyset$. Therefore, $\mathcal{O}_t(U)$ can be written as a disjoint union of the sets each with cardinality $q(q+1)$. Hence the cardinality of $\mathcal{O}_t(U)$ is a multiple of $q(q+1)$. 
\end{proof}

\begin{lemma}\label{lemma7}
	Let $U$ be a subspace in $\mathcal{P}_q(n)$. Then the subspace generated by all the cyclic shifts of $\mathbb{F}_{q^2}$ contained in $U$ is of the form 
	\[W=u_1\mathbb{F}_{q^2}\oplus u_2\mathbb{F}_{q^2}\oplus\cdots \oplus u_r\mathbb{F}_{q^2}~, \]
	where $u_1,u_2,\ldots,u_r \in U$, and we take $r=0$ when $W = \{0\}$. Further, the total number of distinct cyclic shifts of $\mathbb{F}_{q^2}$ contained in  $U$ is $\frac{q^{2r}-1}{q^2-1}$ . 
\end{lemma}

\begin{proof}
	We construct the subspace generated by all the cyclic shifts of $\mathbb{F}_{q^2}$ contained in $U$ as follows. If $U$ has no cyclic shift of $\mathbb{F}_{q^2}$ then we are done. Therefore, let us assume that $U$ contains at least one cyclic shift of $\mathbb{F}_{q^2}$. Let $u_1\mathbb{F}_{q^2}$ be a cyclic shift of $\mathbb{F}_{q^2}$ in $U$. 
Let $u_2\mathbb{F}_{q^2}$ be a cyclic shift of $\mathbb{F}_{q^2}$ in $U$ such that $u_1\mathbb{F}_{q^2}\neq u_2\mathbb{F}_{q^2}$. Suppose that there exists $y \ne 0$ such that $y\in u_1\mathbb{F}_{q^2}\cap u_2\mathbb{F}_{q^2}$. Then $y=u_1z_1=u_2z_2$ for some $z_1,z_2\in \mathbb{F}_{q^2}^*$. This implies that $u_2= u_1z_1z_2^{-1}$, and thus $u_2 \in u_1\mathbb{F}_{q^2}$. This gives $u_2\mathbb{F}_{q^2}\subseteq u_1\mathbb{F}_{q^2}$. As $|u_1\mathbb{F}_{q^2}| = |u_1\mathbb{F}_{q^2}|$, it follows that $u_1\mathbb{F}_{q^2}=u_2\mathbb{F}_{q^2}$, which is a contradiction. Hence $u_1 \mathbb{F}_q^2\cap u_2\mathbb{F}_q^2=\{0\}$. 
	
	 Now let $u_3\mathbb{F}_{q^2}\subseteq U$ be such that $u_3\mathbb{F}_{q^2}\nsubseteq (u_1\mathbb{F}_{q^2}\oplus u_2\mathbb{F}_{q^2})$. Suppose there exists $w \ne 0$ such that $w\in u_3\mathbb{F}_{q^2}\cap  \big(u_1 \mathbb{F}_q^2\oplus u_2\mathbb{F}_q^2\big)$, and let $\beta \in \mathbb{F}_{q^2}\backslash \mathbb{F}_q$. Then $w$ and $\beta w$ both are in $u_3 \mathbb{F}_q^2$, and the subspace $\langle w, \beta w\rangle_{\mathbb{F}_q}$ of $u_3 \mathbb{F}_q^2$ generated by $w$ and $\beta w$ has dimension $2$ over $\mathbb{F}_q$. Therefore, $\langle w,\beta w\rangle_{\mathbb{F}_q} = u_3\mathbb{F}_{q^2}$.
	 
	 As $w\in u_1 \mathbb{F}_q^2\oplus u_2\mathbb{F}_q^2$, we have $w= u_1 x_1+u_2 x_2$ for some $x_1,x_2\in \mathbb{F}_{q^2}$, and hence $\beta w = \beta u_1x_1+\beta u_2x_2$. As $\beta x_1, \beta x_2\in \mathbb{F}_{q^2}$, we get $\beta w\in u_1\mathbb{F}_{q^2}\oplus u_2\mathbb{F}_{q^2}$. Therefore, $u_3\mathbb{F}_{q^2}=\langle w,\beta w\rangle_{\mathbb{F}_q}\subseteq u_1\mathbb{F}_{q^2}\oplus u_2\mathbb{F}_{q^2}$, which is a contradiction. Hence $u_3\mathbb{F}_{q^2}\cap \big(u_1\mathbb{F}_{q^2}\oplus u_2\mathbb{F}_{q^2}\big)=\{0\}$. Continuing in this way, we get that the subspace generated by all the cyclic shifts of $\mathbb{F}_{q^2}$ contained in $U$ is of the form  
\[W = u_1\mathbb{F}_{q^2}\oplus u_2\mathbb{F}_{q^2}\oplus\cdots\oplus u_r\mathbb{F}_{q^2}~,\]
	 where $u_1,u_2,\ldots, u_r\in U$. Clearly, $|W| = q^{2r}$.

Now we count the number of cyclic shifts of $\mathbb{F}_{q^2}$ in $U$. If $U$ does not contain any cyclic shift of $\mathbb{F}_{q^2}$ then $W=\{0\}$, and it trivially holds that there are $\frac{q^{2r}-1}{q^2-1}$ cyclic shifts of $\mathbb{F}_{q^2}$ in $U$. Now suppose that $U$ contains at least one cyclic shift of $\mathbb{F}_{q^2}$. Let $x\mathbb{F}_{q^2}$ be an arbitrary cyclic shift of $\mathbb{F}_{q^2}$ in $W$.  Then $x=x.1\in x\mathbb{F}_{q^2} \subseteq W$, implying that $x \in W$. 
From this follows that all the cyclic shifts of $\mathbb{F}_{q^2}$ contained in $W$ are generated by non-zero elements of $W$. Thus the cyclic shifts of $\mathbb{F}_{q^2}$ in $U$ are precisely the sets $x\mathbb{F}_{q^2}$, $x \in W\setminus \{0\}$. Now any two cyclic shifts $z_1\mathbb{F}_{q^2}$ and  $z_2\mathbb{F}_{q^2}$ of $\mathbb{F}_{q^2}$ are equal if and only if  $z_1= z_2 x$ for some $x \in \mathbb{F}_{q^2}^*$. As there are $q^{2r}-1$ non-zero elements in $W$ and $q^{2}-1$ non-zero elements in $\mathbb{F}_{q^2}$, we conclude that the total number of cyclic shifts of $\mathbb{F}_{q^2}$ contained in $W$ is $\frac{q^{2r}-1}{q^2-1}$.  Hence the result. 
\end{proof}

\begin{theorem}\label{thm3.12}
	Let $n$ be an even number. Let $U$ be a subspace in $\mathcal{G}_q(n,k)$ such that $U$ generates a full-length orbit. Let $\alpha\in \mathbb{F}_q^n\backslash \mathbb{F}_q$ and let the degree of $\alpha$ over $\mathbb{F}_q$ be $2$. Let the number of distinct cyclic shifts of $\mathbb{F}_{q^2}$ in $U$ be $\frac{q^{2m}-1}{q^2-1}$, $m \ge 0$. Then $\dim_{\mathbb{F}_q}(U\cap \alpha U)=2m$. 	
\end{theorem}

\begin{proof}
	Let $V=U\cap \alpha U$.  By Theorem \ref{thm3.9}, $\mathbb{F}_{q^2}\subseteq \mbox{Stab}(V)\cup\{0\}$. Now, $V$ is a vector space over $\mbox{Stab}(V)\cup\{0\}$. Then $V$ is also a vector space over $\mathbb{F}_{q^2}$. Let $\{v_1, v_2, \ldots, v_s\}$ be a basis of $V$ over $\mathbb{F}_{q^2}$, where we take $s=0$ when the basis is empty, i.e., $V=\{0\}$. Then $V$ can be written as 
	\[V=v_1\mathbb{F}_{q^2}\oplus v_2\mathbb{F}_{q^2}\oplus\cdots\oplus v_s\mathbb{F}_{q^2}~.\]
Clearly, $\dim(V)=2s$.

	Now suppose that $U$ contains $\frac{q^{2m}-1}{q^2-1}$  distinct cyclic shifts of $\mathbb{F}_{q^2}$. Clearly, every cyclic shift of $\mathbb{F}_{q^2}$ contained in $V$ is contained in $U$. We claim that a cyclic shift of $\mathbb{F}_{q^2}$ contained in $U$ is also contained in $V$.  Let $u_1\mathbb{F}_{q^2}\subseteq U$, where $u_1\in U$. Then $\alpha u_1\mathbb{F}_{q^2}\subseteq \alpha U$. As $\mbox{deg}_{\mathbb{F}_q}(\alpha)=2$, by Lemma \ref{lemma3.8}, $\alpha \in \mathbb{F}_{q^2}$. So, we get $\alpha u_1\mathbb{F}_{q^2}= u_1\mathbb{F}_{q^2}$. From this we get $u_1 \mathbb{F}_{q^2}\subseteq U\cap \alpha U=V$. Thus all the cyclic shifts of $\mathbb{F}_{q^2}$ that are contained in $U$ are also contained in $V$. Hence $V$ contains exactly $\frac{q^{2m}-1}{q^2-1}$  distinct cyclic shifts of $\mathbb{F}_{q^2}$. By Lemma \ref{lemma7}, $V$ contains exactly $\frac{q^{2s}-1}{q^2-1}$ cyclic shifts of $\mathbb{F}_{q^2}$. Therefore $s=m$, and hence $\dim(V)=2m$. 
\end{proof}

\begin{theorem}\label{thm3.13}
	Let $n$ be an even number.  Let $U$ be a subspace in $\mathcal{G}_q(n,k)$ such that $U$ generates a full-length orbit. Let the number of distinct cyclic shifts of $\mathbb{F}_{q^2}$ in $U$ be $\frac{q^{2m}-1}{q^2-1}$, $m \ge 0$. Then for $0\leq i\leq k-1$, 
\[\lambda_i= \left\{\begin{array}{cc}
q+r q(q+1) & \mbox{if}~ i=2m~,  \\
s_iq(q+1)  & \mbox{if}~ i\ne 2m~,
\end{array}
\right.
\]
where $r, s_i$ are non-negative integers and $\lambda_i$ denotes the number of subspaces $\alpha U$ in Orb$(U)$ such that $\dim(U\cap \alpha U)=i$.	
\end{theorem}

\begin{proof}
	We have already proved in Theorem \ref{thm3.11} that for an odd integer $i,~0\leq i\leq k-1,~\lambda_i$ is a multiple of $q(q+1)$. Let $G=\{\alpha: \alpha \in \mathbb{F}_{q^n}~\mbox{and deg}_{\mathbb{F}_q}(\alpha)=2\}$. 
	Let $\beta \in G$.  As $U$ contains $\frac{q^{2m}-1}{q^2-1}$ distinct cyclic shifts of $\mathbb{F}_{q^2}$, by Theorem \ref{thm3.12}, $\dim(U\cap \beta U)=2 m$. By Lemma \ref{lemma2.1} and Lemma \ref{lemma3.1}, $\mathcal {S}_{\beta,U}\subseteq \mathcal{O}_{2m}(U)$. The cardinality of $\mathcal{S}_{\beta,U}$ is $q$, by Lemma \ref{lemma3.7}. Now let there exist  $\gamma U\in  \mathcal{O}_{2m}(U)~ \mbox{with}~ \gamma U\notin \mathcal{S}_{\beta, U}$ then the degree of $\gamma$ over $\mathbb{F}_q$ cannot be two, by Lemma \ref{lemma3.8} and Remark \ref{remark3.1}. By Lemma \ref{lemma3.5}, the cardinality of $\mathcal{S}_{\beta,U}$ is $q(q+1)$ and by Lemma \ref{lemma3.4}, $\mathcal{S}_{\gamma,U}\cap \mathcal{S}_{\beta,U}=\emptyset$. Therefore, $\lambda_{2m}=q+rq(q+1)$ for some non-negative integer $r$. 
	
	Let $0\leq i\leq k-1,~i\neq 2m$, and $i$ an even number. Let $\mathcal{O}_i(U)\neq \emptyset$, and let $\alpha U\in \mathcal{O}_i(U)$. As $i\neq 2m$, by Theorem \ref{thm3.12} the degree of $\alpha$ over $\mathbb{F}_q$ is not 2. By Lemma \ref{lemma2.1} and Lemma \ref{lemma3.1}, $\mathcal{S}_{\alpha,U} \subseteq \mathcal{O}_i(U)$.  By Lemma \ref{lemma3.5},  the cardinality of $\mathcal{S}_{\alpha,U}$ is $q(q+1)$. Let there exist $wU\in \mathcal{O}_i(U)$ such that $wU\notin \mathcal{S}_{\alpha,U}$ Then $\mathcal{S}_{w,U}\subseteq \mathcal{O}_i(U)$ and $\lvert \mathcal{S}_{w,U}\rvert= q(q+1)$. Hence $\lambda_i=\lvert\mathcal{O}_i(U)\rvert$ is a multiple of $q(q+1)$, i.e., $\lambda_i = s_iq(q+1)$ for some non-negative integer $s_i$. 
\end{proof}
In the above theorem, if $U$ does not contain a cyclic shift of $\mathbb{F}_q^2$, then $\lambda_0=q+rq(q+1)$, for some non-negative integer $r$, and for $1\leq i\leq k-1,~\lambda_i$ is a multiple of $q(q+1)$. 

\begin{example}
	Let $q=3~\mbox{and}~ n=16$. Let $z$ be a primitive element of $\mathbb{F}_{3^{16}}$. Let $\beta =2z^{15} + 2z^{14} + z^{13} + z^{11} + z^{10} + z^8 + z^7 + z^5 + z^4 + z^3$. The degree of $\beta$ over $\mathbb{F}_3$ is $2$. Then $\beta$ is a primitive element of $\mathbb{F}_{3^2}$. Define $U=\gamma_1\mathbb{F}_{3^2}\oplus \gamma_2\mathbb{F}_{3^2}\oplus \gamma_3 \mathbb{F}_{3^2}\oplus \alpha \mathbb{F}_3$, where $\gamma_1=2z^{15} + z^{10} + 2{z^6} + z^5 + z^4 + z^3 + 2z + 1,~ \gamma_2=z^{14} + 2z^{12} + z^5 + 2z^4 + z^2 + 1,~ \gamma_3=z^5+1~\mbox{and}~ \alpha=z^4+z^3+1$. Then $U$ is a vector space of dimension $7$ over $\mathbb{F}_3$, and $U$ contains $\frac{3^6-1}{3^2-1}$ distinct cyclic shifts of $\mathbb{F}_{3^2}$. As $\mbox{gcd}\big(\dim(U), n\big)=1,~ U$ generates a full-length orbit. By doing computation using magma we get $\lambda_0=20438244,~\lambda_1=1050192,~ \lambda_2=34632,~\lambda_3=288,~\lambda_4=0,~\lambda_5=0~\mbox{and}~\lambda_6=3$. As the number of cyclic shifts of $\mathbb{F}_{q^2}$ contained in $U$ is $\frac{q^6-1}{q^2-1}$, $\lambda_6=3$ and other non-zero $\lambda_i$ is a multiple of $q(q+1)=12$. 
\end{example} 

Now we consider the subspaces of $\mathbb{F}_{q^n}$ which do not generate a full-length orbit. It is known that such a subspace has a stabilizer $\mathbb{F}_{q^t}^*$, where $t~(> 1)$ is a divisor of $\mbox{gcd}(\dim_{\mathbb{F}_q}(U),n)$. The following theorem  generalizes the idea of both Theorems \ref{thm3.7} and \ref{thm3.13} in the present context. 
\begin{theorem}\label{thm3.14}
	Let $U$ be a subspace in $\mathcal{G}_q(n,k)$ such that $U$ does not generate a full-length orbit. Let the stabilizer of $U$ be $\mathbb{F}_{q^t}^*$.  
	\begin{enumerate}
		\item If $\frac{n}{t}$ is an odd number then $\lambda_i$ is a multiple of $q^t(q^t+1)$, $0 \le i \le k-1$.
		\item If $\frac{n}{t}$ is an even number and $U$ contains $\frac{q^{2tm}-1}{q^{2t}-1}~(m\geq0)$ number of distinct cyclic shifts of $\mathbb{F}_{q^{2t}}$ then 
		\[\lambda_{2tm}=q^t+ s q^t(q^t+1)~,\] for some non-negative integer $s$, and for $0\leq i\leq k-1, ~ i\neq 2tm$, and $\lambda_i$ is a multiple of $q^t(q^t+1)$. 
			\end{enumerate}
\end{theorem}

\begin{proof}
	As $\mbox{Stab}(U)= \mathbb{F}_{q^t}^*$, $U$ and all its cyclic shifts are vector spaces over the field $\mathbb{F}_{q^t}$ and $t$ divides  $\gcd\big(\dim_{\mathbb{F}_q}(U),n\big)$. For any $\alpha \in \mathbb{F}_{q^n}\backslash \mathbb{F}_{q^t},~ U\cap \alpha U$ is a vector space over the field $\mathbb{F}_{q^t}$.
	
	For $0\leq i\leq k-1$, $\lambda_i$ denotes the number of subspaces $\alpha U$ such that  $\dim_{\mathbb{F}_q}(U\cap \alpha U)=i$. For $0\leq j\leq \frac{k}{t} -1$, let  $\chi_j$ denotes the number of subspaces $\alpha U$ such that $\dim_{\mathbb{F}_{q^t}}(U\cap \alpha U)=j$.   It is easy to check that $\dim_{\mathbb{F}_q}(U\cap\alpha U)= t\dim_{\mathbb{F}_{q^t}}(U\cap \alpha U)$. Thus the dimension of $U\cap \alpha U$ over $\mathbb{F}_q$ is a multiple of $t$.  Hence $\lambda _i=0$ if $t$ does not divide $i$, and for $0\leq j\leq \frac{k}{t}-1,~ \lambda_{jt}=\chi_j$.
	
	Let $n_1= \frac{n}{t}$, and  let $q_1 = q^t$. Then $q^n= q_1^{n_1}$.
	Now $U\subseteq \mathbb{F}_{{q_1}^{n_1}}$ is a vector space over $\mathbb{F}_{q_1}$ and $\dim_{\mathbb{F}_{q_1}}(U)=\frac{k}{t}$. Now $U$ generates a full-length orbit as a vector space over $\mathbb{F}_{q1}$, i.e., $\lvert \mbox{Orb}(U)\rvert = \frac{{q_1}^{n_1}-1}{q_1-1}$.  \\
	\textbf{Case 1}. Let $n_1=\frac{n}{t}$ be an odd number.\\
	By Theorem \ref{thm3.7} , for each $j,~0\leq j\leq \frac{k}{t}-1,~\chi_j$  is a multiple of $q_1(q_1+1)$. Therefore, for each  $j,~ 0\leq j\leq \frac{k}{t}-1,~ \lambda_{jt}$ is a multiple of $q_1(q_1+1)$. As $q_1=q^t$,  $\lambda_{jt}$ is a multiple of $q^t(q^t+1)$.\\ 
	\textbf{Case 2}. Let $n_1=\frac{n}{t}$ be an even number.\\
	By Theorem \ref{thm3.13}, if $U$ contains $\frac{q_1^{2m}-1}{q_1^2-1}$ number of distinct cyclic shift of $\mathbb{F}_{{q_1}^2}$ then $\chi_{2m}= q_1+ sq_1(q_1+1)$, for some non-negative integer $s$. For $0\leq j\leq \frac{k}{t}-1,~ j\neq 2m,~ \chi_j$ is a multiple of $q_1(q_1+1)$. From this we conclude that if $U$ contains $\frac{q^{2tm}-1}{q^{2t}-1}$ cyclic shifts of $\mathbb{F}_{q^{2t}}$ then $\lambda_{2tm}=q^{t}+sq^t(q^t+1)$, and for $0\leq j\leq \frac{k}{t}-1,~j\neq 2m,~ \lambda_{jt}$ is a multiple of $q^t(q^t+1)$. 
\end{proof}

\begin{example}
	Let $q=2$ and $n=14$. Let $z$ be a primitive element of $\mathbb{F}_{2^{14}}$. Then $z^{\frac{2^{14}-1}{2^2-1}}=z^{5461}$ is a primitive element of $\mathbb{F}_{2^2}$. Now, $\mathbb{F}_{2^2}=\{0,1, z^{5461}, z^{10922}\}$. Define $U=z^{11}\mathbb{F}_{2^2}\oplus z^{13} \mathbb{F}_{2^2}\oplus z^{14} \mathbb{F}_{2^2}$. The subspace $U$ is a vector space over $\mathbb{F}_2$ of dimension $6$. Clearly, $\mathbb{F}_{2^2}$ is a stabilizer of $U$ and $\frac{n}{t}= 7$ is an odd number. By doing computation using Magma we get $\lambda_0= 5040~, \lambda_1=0, ~ \lambda_2= 420,~ \lambda_3=\lambda_4=\lambda_5=0$. We can see that $\lambda_i$ is $0$ if $2$ does not divide $i$, and non-zero $\lambda_i$ is a multiple of $2^2 (2^2+1)=20$. 
\end{example}

\begin{example}
	Let $q=3$ and $n=12$. Let $z$ be a primitive element of $\mathbb{F}_{3^{12}}$. Let $a\in \mathbb{F}_{3^4}\backslash \mathbb{F}_{3^2}~\mbox{and}~ \gamma_1,\gamma_2\in \mathbb{F}_{3^{12}}\backslash \mathbb{F}_{3^4}$. Define $U=a\gamma_1\mathbb{F}_{3^2}\oplus a^2\gamma_1\mathbb{F}_{3^2}\oplus\gamma_2\mathbb{F}_{3^2}$. Let $a=2z^{10} + 2z^{9}+ 2z^7 + z^6 + z^4 + 2z^3 + z + 2,~\gamma_1=z^2+1~\mbox{and}~\gamma_2=z^3+z+1$. The subspace $U$ is a vector space over $\mathbb{F}_3$ and $\dim_{\mathbb{F}_3}(U)=6$. The stabilizer of $U$ is $\mathbb{F}_{3^2}~\mbox{and}~ \frac{n}{t}= 6$ is an even number. As $a\in \mathbb{F}_{3^4}\backslash \mathbb{F}_{3^2},~ a\gamma_1\mathbb{F}_{3^2}\oplus a^2\gamma_1\mathbb{F}_{3^2}=\gamma_1\mathbb{F}_{3^4}$. Thus, $U$ contains a cyclic shift of $\mathbb{F}_{3^4}$. By doing computation using Magma we get $\lambda_0=58320,~\lambda_1=0,~\lambda_2=8100~, \lambda_3=0,~ \lambda_4=9~ \mbox{and}~\lambda_5=0$. We can see that $\lambda_i=0$ if $t=2$ does not divide $i$. As $U$ contains $1$ cyclic shift of $\mathbb{F}_{3^4}$, $\lambda_4= 9$ and other non- zero $\lambda_i$ is a multiple of $q^2(q^2+1)=90$. 
	
\end{example}

\section{Conclusion}
In this paper, we have obtained some results on the intersection distribution and hence the distance distribution of single-orbit cyclic subspace codes.  
For a subspace $U$ in $\mathcal{G}_q(n,k)$ with $\mbox{Stab}(U)=\mathbb{F}_{q^t}^*,~t\geq 1$, and the distance distribution $(\delta_{2k-2l}, \delta_{2k-2(l-1)},\ldots, \delta_{2k})$, we have proved that if  $\frac{n}{t}$ is an odd number, then for  $0\leq i\leq l$, 
\[\delta_{2k-2i}= r_iq^t(q^t+1)~,\] 
where $r_i$ is some non-negative integer.
Further, if $\frac{n}{t}$ is an even number and $U$ contains $ \frac{q^{2tm}-1}{q^{2t}-1}$ number of distinct cyclic shifts of $\mathbb{F}_{q^{2t}}$, then for $0\leq i\leq l,~i\neq 2tm$,
\[\delta_{2k-2i}=s_iq^t(q^t+1)~,\] 
and
\[\delta_{2k-2(2tm)}=q^t+s_{2tm}q^t(q^t+1)~,\] 
where $s_i$ is some non-negative integer.

This area requires more investigation. It is not known what values the integers $r_i$ and $s_i$ can take. Further, the distance distribution of cyclic subspace codes with multiple orbits can be studied.

\bmhead{Data Availability}
No data was used for the research described in this article.

\bmhead{Conflicts of interest}
The authors have no conflicts of interest to declare that are relevant to the content of this article.

\bmhead{Acknowledgments}
The first author would like to thank Ministry of Human Resource Development (MHRD), India for providing financial support.
\bibliography{sn-bib}


\end{document}